\documentclass[a4paper,11pt]{amsart}

\usepackage[latin1]{inputenc}
\usepackage{amsmath}
\usepackage{amsthm}
\usepackage{amscd}
\usepackage{amssymb}
\usepackage{amsfonts}
\usepackage{pb-diagram, pb-xy}
\usepackage[all]{xy}
\usepackage{hyperref}
\usepackage{cleveref}
\usepackage{MnSymbol}
\hypersetup{colorlinks=true}
\usepackage{float}

\newcommand{\CC}{\mathcal{C}}
\newcommand{\F}{\mathbb{F}}

\newcommand{\N} {\ensuremath{{\rm N}}}

\newcommand{\Z} {\ensuremath{\mathbb{Z}}}
\newcommand{\wt} {\ensuremath{\textnormal{wt}}}

\newcommand{\ho}{\mbox{\rm Hom}}

\newtheorem*{thm*}{Theorem}
\newtheorem{thm}{Theorem}[section]
\newtheorem{dfn}[thm]{Definition}

\newtheorem{cor}[thm]{Corollary}
\newtheorem{rmk}[thm]{Remark}

\newtheorem{exa}[thm]{Example}

\title{On the algebraic structure of quasi group codes}
\author{Martino Borello}
\thanks{M. Borello is with Universit\'e Paris 8, Laboratoire de G\'eom\'etrie, Analyse et Applications, LAGA,
Universit\'e Sorbonne Paris Nord, CNRS, UMR 7539, France}

\author{Wolfgang Willems}
\thanks{W. Willems is with Otto-von-Guericke Universit\"{a}t, Magdeburg, Germany and
Universidad del Norte, Barranquilla, Colombia}

\begin{document}

\begin{abstract}
In this note, an intrinsic description of some families of linear
codes with symmetries is given, showing that they can be described
more generally as quasi group codes, that is, as linear codes allowing a
group of permutation automorphisms which acts freely on the set of coordinates. An algebraic description,
including the concatenated structure, of such codes is presented. This allows to construct quasi group codes from codes over rings, and vice versa.
The last part of the paper is dedicated to the investigation of self-duality of quasi group codes.
\end{abstract}

\maketitle

\section{Introduction}

In the theory of error correcting codes linear codes play a
central role due to their algebraic structure which allows, for
example, an easy description and storage. Quite early in coding theory
 it appeared convenient to add additional algebraic structure
in order to get more information about the parameters and to speed up
 the decoding process. In 1957, E. Prange introduced the
now well-known class of cyclic codes \cite{P}, which are the
forefathers of many other families of codes with symmetries
discovered thereafter. In particular, abelian codes \cite{B}, group codes
\cite{M}, quasi-cyclic codes \cite{CPW}, quasi-abelian codes
\cite{W}, and twisted group codes \cite{DW} are distinguished descendants of cyclic codes. All of them have 
nice algebraic structures, and many remarkable and optimal codes belong to these
families. Moreover it is proved that the family of group codes is asymptotically good \cite{BM,BMS,BW2}, as the family of quasi-cyclic and quasi-abelian codes \cite{BGSS,JL,K}.

The aim of the current  note is to give an intrinsic description of
these classes in terms of their permutation automorphism groups and to deduce some structural properties from that. We
show that all of them admit a subgroup of the full group of permutation
automorphisms which acts freely on the set of coordinates. A linear code with this
property will be called a quasi group code. So all families cited above can be
described in this way. Note that quasi group codes (over Frobenius
rings) have also been introduced  in a recent paper by S. Dougherty
et al. \cite{DGTT}. It turns out  that the families of quasi group codes, quasi-abelian codes
and quasi-cyclic codes coincide. Furthermore, we give the most general
algebraic description  of quasi group codes (Theorem \ref{thm:str}) and describe
their concatenated structure (Theorem \ref{thm:con})
without any restriction on the group which acts on. In addition, no assumption
on semi-simplicity (unlike in previous papers on the subject) is
needed. In contrast with the concatened structure of quasi-abelian codes, quasi group codes for a non-abelian group can be decomposed as concatenations of codes over non-necessary commutative (Frobenius) rings, such as matrix rings. There is an active research on codes over (Frobenius) rings over the last decades (see for example \cite{AS,ASSY,DS}). In particular, codes over matrix rings may be used in applications to Space-Time Coded Modulations \cite{OSB}. From the investigation of good quasi group codes we get, as a byproduct, constructions of good codes over (Frobenius) rings, and vice versa. 

In the last part of the paper, we deal with self-duality of
quasi group codes. We prove a necessary and sufficient condition for the existence of self-dual quasi group codes depending
on the underlying field and on the index.

\bigskip

\section{Background}\label{sec:bac}

In this section we collect  some preliminaries which are crucial in the subsequent sections.\\

Let $K$ be a finite field of cardinality $q$.
 A \emph{linear code} $\CC$ of \emph{length} $n$ is a
$K$-linear subspace of $K^n$. An element $c=(c_1,\ldots,c_n)\in\CC$ is called a
\emph{codeword} and its (Hamming) \emph{weight} is given by
$$\wt(c):= |\{i\in\{1,\ldots,n\} \ | \ c_i\neq 0\}|.$$
The \emph{minimum distance} of $\CC$ is defined by  ${\rm d}(\CC):=
\min_{c\in \CC \setminus\{0\}} \wt(c)$. An $[n,k,d]_q$ code is a linear code
of length $n$, dimension $k$ and minimum distance $d$ over a field
of cardinality $q$.
 These, i.e.,  $n,k,d$ and $q$ are usually called \emph{the parameters} of the code.\\

There is a standard way of combining codes to obtain a code of
larger length. To describe this process, let
 $K\subseteq L$ be a field extension, $m$ an integer greater than or equal to $[L:K]$ and let $\pi:L\to K^m$ be a $K$-linear injection.
In the concatenation process, $L$-linear codes in $L^n$ are called
\emph{outer codes} and the $K$-linear code $\mathcal{I}:=\pi(L)$ is
called an \emph{inner code}. If $C\subseteq L^n$ is an $L$-linear code, then the $K$-linear code
$$\mathcal{I} \ \Box_\pi \ \mathcal{C}:=\{(\pi(c_1),\ldots,\pi(c_n))\ | \ (c_1,\ldots,c_n)\in \mathcal{C}\}\subseteq K^{n\cdot m}$$
is called the \emph{concatenation of} $\mathcal{C}$ \emph{with}
$\mathcal{I}$ \emph{by} $\pi$, or simply the \emph{concatenated
code}.

This construction obviously depends on the choice of $\pi$, but
there are some properties independent of  $\pi$. For example,
the length of $\mathcal{I} \ \Box_\pi \ \mathcal{C}$ is $n\cdot m$,
the $K$-dimension  of $\mathcal{I} \ \Box_\pi \ \mathcal{C}$ is
the product $\dim_K(\mathcal{I})\cdot \dim_L(\mathcal{C})$, and we
have the bound
 ${\rm d}(\mathcal{I} \ \Box_\pi \ \mathcal{C})\geq {\rm d}(\mathcal{I})\cdot
{\rm d}(\mathcal{C})$ (see for example \cite{HP}).\\

There are some natural group actions associated to linear codes.
The symmetric group $S_n$ acts on the set $\{1,\ldots,n\}$ of coordinates
by definition. This action induces an \emph{action on the elements} of
$K^n$, namely  for $v\in K^n$ and $\sigma\in S_n$ we have
$$v^\sigma:=(v_{\sigma^{-1}(1)},v_{\sigma^{-1}(2)},\ldots,v_{\sigma^{-1}(n)}),$$
which induces an
\emph{action on subsets} of $K^n$ (and in particular on linear codes). For
$\mathcal{C}\subseteq K^n$ we put
$$\mathcal{C}^\sigma:=\{c^\sigma \ | \ c\in \mathcal{C}\}.$$
An element $\sigma\in S_n$ is an  \emph{automorphism} of
$\mathcal{C}$ if $\mathcal{C}^\sigma=\mathcal{C}$.  The stabilizer
$${\rm PAut}(\mathcal{C}):=\{\sigma\in S_n \ | \
\mathcal{C}^\sigma=\mathcal{C}\}$$ is called the
\emph{permutation automorphism group} of $\mathcal{C}$.
Moreover, a linear code $\mathcal{C}_1$ is  \emph{equivalent} or better \emph{permutation equivalent} to
a linear code $\mathcal{C}_2$ 
if there exists $\sigma\in S_n$ such that $\mathcal{C}_1^\sigma=\mathcal{C}_2$.
This is not the most general definition of equivalence but it is sufficient for our purpose.
Finally, it is easy to see that ${\rm PAut}(\CC^\sigma)={\rm PAut}(\CC)^\sigma$ ({\it the conjugate of} ${\rm PAut}(\CC)$ {\it in} $S_n$ 
{\it by} $\sigma$).\\

From group theory we recall that a (right) action of a group $G$ on a set $X$ is
called
\begin{itemize}
\item \emph{transitive} if $X\neq \emptyset$ and for all $x,y\in
X$ there  exists $g \in G$ such that $x^g=y$;
\item \emph{free} if $x^g=x$ for some $x \in X$ and $g \in G$, then $g$ is the unit in $G$;
\item \emph{regular} if it is both transitive and free.
\end{itemize}

In the case that the group $G$ is finite and acts freely, it immediately
follows   that all orbits have
cardinality $|G|$. In particular, $|G|$ divides $|X|$ and $\ell = \frac{|X|}{|G|}$ is called the {\it index} of the action. 
Moreover, if $G$ is
regular, then $|G|= |X|$. \\

A subgroup $H$ of the symmetric group $S_n$ is called
{\it transitive} (resp. {\it regular}) if  of the action of $H$
 on the set of coordinates
$\{1,\ldots,n\}$ is transitive (resp. regular).\\

Finally, given a group $G$ and a field $K$, the \emph{group algebra}
$KG$ is the set of formal sums $$KG:=\left\{\left.a=\sum_{g\in
G}a_gg \ \right| \ a_g\in K\right\},$$ which is a $K$-vector space
in a natural way and which becomes a $K$-algebra via the
multiplication
 $$ab:=\sum_{g \in G}\left(\sum_{h\in G} a_hb_{h^{-1}g}\right)g,$$
for $\displaystyle a=\sum_{g\in G}a_gg$ and $\displaystyle
b=\sum_{g\in G}b_gg$.

\bigskip

\section{Codes with a free acting group of symmetries}

Let $G$ be a finite group of cardinality $n$ and let $K$ be a finite field.
Recall that the group algebra $KG$ is isomorphic to $K^n$ as a $K$-vector space, where $n=|G|$. There is a standard
way of constructing such an isomorphism, which allows us to transfer many
coding theoretical properties from $K^n$ to $KG$. Once an
ordering $g_1,\ldots,g_n$ of the elements of $G$ is chosen, we may define
$\varphi: g_i\mapsto e_i$, where $\{e_1,\ldots,e_n\}$ is the
standard basis of $K^n$. Then we extend this map $K$-linearly so that
\begin{equation}\label{eq:iso}
\varphi: \sum_{i=1}^na_ig_i\mapsto (a_1,\ldots,a_n).
\end{equation}
The isomorphism $\varphi$ obtained in this way is not canonical, since it
depends on the ordering of the group. But different orderings lead only to a  permutation of the coordinates, hence to permutation equivalent codes.

Via the isomorphism $\varphi$, we may transfer the Hamming metric
from $K^n$ to $KG$. For $a\in KG$, we define
$\wt(a):=\wt(\varphi(a))$. So, from a coding theoretical point of
view, we can consider linear codes either in $KG$ or in $K^n$
without any difference. However, the algebraic structure of $KG$
allows us to consider codes with more structure than linearity.

\begin{dfn}[\cite{B,M}] {\rm
A {\it $G$-code}  is a right ideal
$\CC$ in the group algebra $KG$. If the group $G$ is cyclic (resp. abelian), then
the code $\CC$ is called a {\it cyclic {\rm(resp.} abelian{\rm)} $G$-code}.
In the case we do not specialize the group $G$ explicitly we briefly speak of a group code.
}
\end{dfn}

The restriction to right ideals is only for convention, which means that everything in the following may be stated equally  for left ideals.\\

The particular class of cyclic $G$-codes  is nothing else than the family of well known
cyclic codes. If $G$ is cyclic, hence generated by a certain $g \in G$, and the
isomorphism $\varphi$ sends $g^i\to e_{i+1}$ for all
$i\in\{0,\ldots,n-1\}$, then $\varphi(\CC)$ is \emph{cyclic}. In
fact, in this case $KG\cong K[x]/(x^n-1)$ via the map
$g\mapsto x +(x^n-1)$. Thus, a cyclic code turns out to be an ideal in the factor algebra $K[x]/(x^n-1)$,
 which is the classical definition. \\

Now let $\CC$ be a $G$-code with $n=|G|$. Observe that 
the right multiplication on $G$ by one of its elements, say $g$, induces a
permutation $\sigma_g\in S_n$ defined by
\begin{equation}\label{eq:act}
\sigma_g(i)=j \ \text{ iff } \ g_ig=g_j.
\end{equation}
Note that $g\mapsto \sigma_g$ is a faithful permutation
representation of $G$, which  depends on the chosen ordering of
$G$. If this map coincides with that chosen for $\varphi$, then
$\sigma(G):=\{\sigma_g \ | \ g\in G\}$ is a subgroup of ${\rm
PAut}(\varphi(\CC))$. This is due to the fact that a right ideal is
stable by multiplication on the right. Since the action of right
multiplication is regular, $\sigma(G)$ is a regular subgroup of $S_n$. \\

Suppose that $\CC$ is a linear code in $K^n$ admitting a regular
subgroup $G$ of ${\rm PAut}(\CC)$. Since $G$ is a group of
automorphisms, $\mathcal{C}$ becomes a right $KG$-module via the action
\begin{equation}\label{eq:prod}
c\cdot\left(\sum_{g\in G}a_gg\right):=\sum_{g\in G}a_gc^g
\end{equation}
for $c\in \mathcal{C}$ and $a_g\in K$, where $c^g \in \CC$ is the image of $c$ under the action of $g$. Moreover, as every regular
action of $G$ is isomorphic to the action of $G$ on itself given by
right multiplication, there is an ordering of $G$ such that
$\varphi^{-1}(\CC)$ is a $G$-code in $KG$. 
Thus we have proved, in our framework, the known characterization of
group codes.

\begin{thm}[\cite{BdRS}] \label{BRS}
Let $G$ be a group of order $n$ and let $\mathcal{C}$ be a linear code
in $K^n$. Then $\mathcal{C}$ is a $G$-code if and only if $G$ is
isomorphic to a regular subgroup $H$ of ${\rm PAut}(\mathcal{C})$.
\end{thm}

We would like to mention here that a $G$-code may also be an $H$-code where $H$ is not isomorphic to $G$. For instance, the
binary extended [24,12,8] Golay code is a $G$-code for the symmetric group $S_4$ \cite{BLM} and the dihedral group $D_{24}$  \cite{LH}.
Furthermore, there are abelian $G$-codes which are not group codes for cyclic groups. 
As an example may serve 
 the binary extended $[8,4,4]$ Hamming code. It is a $G=C_2\times C_4$ code, but not equivalent to a cyclic code (in fact, its automorphism group is isomorphic to ${\rm AGL}_3(2)$ that does not contain an element of order $8$). 

\begin{dfn} {\rm
Let $G$ be a finite group. A $K$-linear code $\CC$ is called a {\it quasi-$G$ code of index $\ell$} if $\CC$ is a right $KG$-submodule of
$KG^\ell = KG \oplus \cdots \oplus KG$ ($\ell$-times) for some  $\ell \in \N$. A quasi group code is a quasi-$G$ code for some group $G$.
In the case that $G$ is  cyclic (resp. abelian) we call $\CC$ a {\it cyclic} (resp. {\it abelian}) {\it quasi group code}.}
\end{dfn}

Clearly, in the case $\ell =1$, a quasi $G$-code is a $G$-code, just by definition. The trivial quasi-$G$ codes over $K$ (i.e., $G=1$) are
nothing else than the linear codes over $K$.

\begin{rmk} \label{R1} {\rm If $\CC$ is a $G$-code, then $\CC$ is a quasi-$H$ code of index $\ell = \frac{|G|}{|H|}$ for any subgroup $H$ of $G$.
This follows immediately from the fact that $KG = \oplus_{t \in T} tKH$ if $T$ is a left transversal of $H$ in $G$.}
\end{rmk}

\begin{thm} Let $G$ be a group of order $|G|=\frac{n}{\ell}$ and let $\CC$ be a linear code in $K^n$. Then 
$\CC$ is a quasi-$G$ code of index $\ell$ if and only if $G$ is isomorphic to a subgroup $H$ of ${\rm PAut}(\mathcal{C})$ which acts freely
of index $\ell$ on the coordinates.
\end{thm}
\begin{proof} Suppose that $C$ is a right $KG$-submodule of $A=KG^\ell$. Note that $G$ acts regularly on the index set $\{g \in G\}$ of each component
$KG$. Hence $G$ acts freely of index $\ell$ on the set of all coordinates of $A$. Furthermore, the matrix group
$P(G) = \{P(g) \mid g \in G\}$ induced by the right action of $G$ on $A$ leaves $\CC$ invariant since $\CC$ is a right $KG$-submodule of $A$.
Hence $G \cong P(G)$ is a subgroup of ${\rm PAut}(\CC)$ which acts freely of index $\ell$ on the coordinates.

Now suppose that  $G$ is isomorphic to a subgroup $H$ of ${\rm PAut}(\mathcal{C})$ which acts freely
of index $\ell$ on the coordinates. Let $P:G \longrightarrow H$ be an isomorphism from
$G$ to $H$. We define an action of $G$ on $\{1, \ldots, n\}$ by $ig=j$ if and only if $e_i P(G) = e_j$ where the $e_i$ is the standard basis of 
$K^n$. Since $P(G)=H$ acts freely of index $\ell$ on the coordinates, $G$ has exactly $l$ orbits ${\mathcal O}_j$ of length $|G|$ on
$\{1, \ldots, n\}$. Next we fix representatives $i_1, \ldots, i_\ell$ of ${\mathcal O}_1, \ldots, {\mathcal O}_\ell$ and identify
 $e_r \in {\mathcal O}_j$ with $g \in G$ if $i_jg = r$. Thus
$$ \oplus_{r \in {\mathcal O}_j}Ke_r = KG$$
 as a right $KG$-module and
$$K^n = \oplus_{j=1}^\ell(\oplus_{r \in {\mathcal O}_j}Ke_r) = \oplus_{j=1}^\ell KG = KG^\ell.$$
Since $H$ is a subgroup of ${\rm PAut}(\CC)$ the code $\CC$ is a right $KG$-submodule of $KG^\ell$.
\end{proof}

Note that a cyclic quasi group code of index $\ell$ is nothing else than a quasi-cyclic code in the classical sence of index $\ell$.
Thus, by Remark \ref{R1}, a nontrivial quasi group code is a quasi-cyclic code. 

\begin{cor} The class of nontrivial group codes of length $n>1$ coincides with the class of quasi-cyclic codes of length $n>1$. \end{cor}

\begin{exa} {\rm
The extended binary Golay code of length $24$ can be seen in many
different ways. For example, it is a quasi group code for $D_6$ of
index $4$ (see Example \ref{exa:golay}), but also a quasi group code for $A_4$ of index $2$, and
even a group code for $S_4$, $D_{24}$ $C_3\times D_8$, $C_2\times
A_4$ and $(C_6\times C_2)\rtimes C_2$ \cite{BLM,DGTT}. }
\end{exa}

\begin{rmk} {\rm
A natural question is the following. What can we say about the code if a group
does not act freely? In this case, the situation gets more complicated
to be treated in a general framework, since there are many
possible configurations of the fixed points of the automorphisms
which give rise to different module structures. Some results in this
direction can be found in \cite{BW} for very small groups and in the case of self-dual
codes.}
\end{rmk}

\bigskip

\section{The concatenated structure of quasi group codes}

It it well-known (see \cite[Chapter VII, \S 12]{HB}) that every
group algebra $KG$ can be uniquely  decomposed (up to a permutation of the
components) into a direct sum of indecomposable two-sided ideals as
\begin{equation}\label{eq:decomposition}
KG=\underbrace{f_0KG}_{\mathcal{B}_0}\oplus \ \cdots \ \oplus \underbrace{f_sKG}_{\mathcal{B}_s},
\end{equation}
where the $\mathcal{B}_i$'s are called \emph{blocks} and the $f_i$'s are primitive orthogonal idempotents in the center of $KG$ with
$1 = f_0+ \cdots + f_s$. Note that each $\mathcal{B}_i$ is a Frobenius ring. 
 The image $\varphi(\mathcal{B}_i)$ of ${\mathcal B}_i$ under the map $\varphi$ defined as in \eqref{eq:iso} is a linear code of length $|G|$ over $K$, 
which is uniquely determined by $G$ (up to equivalence, from the choice of $\varphi$).\\
Moreover, for every $KG$-module $\CC$, we have a ``blockwise'' direct decomposition
\begin{equation}\label{eq:dec}\CC=\underbrace{\CC f_0}_{\CC_0}\oplus \ \cdots \ \oplus \underbrace{\CC f_s}_{\CC_s},\end{equation}
of $\CC$ into $KG$-modules $\CC_i =\CC f_i$.
Observe that $\CC_i$ is indeed a $KG$-module since $f_i$ lies in the center of $KG$.
Clearly, $\CC_i$ is a $\mathcal{B}_i$-module too.\\

Now let $G$ be a subgroup of $S_n$ which acts freely on $\{1,\ldots, n\}$. Let $m=|G|$ and $n=m\ell$. Then
$$\varphi^\ell: KG^\ell\to K^n,$$
where $\varphi$ is defined as in \eqref{eq:iso}, is an isomorphism
of vector spaces. This map can be restricted to $\mathcal{B}_i^\ell$
for every $i$, getting a $\mathcal{B}_i$-module isomorphism between
$\mathcal{B}_i^\ell$ and the $\mathcal{B}_i$-modules in $K^n$. 

\begin{thm}\label{thm:str}
Every quasi-$G$ code $\CC$ can be blockwise decomposed as
$$\CC=\CC_0\oplus \ \cdots \ \oplus \CC_s,$$
 where each $\CC_i$ is a linear code
(eventually trivial) of length $\ell$ over the ring $\mathcal{B}_i$,
i.e. a $\mathcal{B}_i$-submodule of $\mathcal{B}_i^\ell$.
\end{thm}
\begin{proof} By definition,  $\CC$ is a right $KG$-submodule of $KG^\ell= KG \oplus \cdots \oplus KG$.
According to \eqref{eq:dec} there is a decomposition
$$\CC = \CC f_0 \oplus \cdots \oplus \CC f_s = \CC_0 \oplus \cdots \oplus \CC_s$$
where
$$\CC_i = \CC f_i \leq (KG \oplus \cdots \oplus KG)f_i \leq  KG f_i \oplus \cdots \oplus KG f_i = {\mathcal B}_i^\ell.$$
Thus $\CC_i$ is a code over ${\mathcal B}_i$ of length $\ell$. Clearly this code is linear over ${\mathcal B}_i$ (from the right) 
since $C_i$ is a right $KG$-module.
\end{proof}

For each $ i \in \{1,\ldots, \ell\}$ we fix a $K$-linear injection
$$\pi_i:\mathcal{B}_i\longrightarrow K^m.$$ Similar to the concatenation in Section \ref{sec:bac}
$${\mathcal B}_i \ \Box_{\pi_i} \ \mathcal{C}_i:=\{(\pi_i(c_1),\ldots,\pi_i(c_l))\ | \ (c_1,\ldots,c_l)\in 
\mathcal{C}_i\}\subseteq K^{m\ell}=K^n.$$
This is a linear code over $K$ which we also call  a
\emph{concatenated code}.\\

With the previous notations we have the following. 

\begin{thm}\label{thm:con} 
{\rm a)}
Every quasi-$G$ code $\CC$ can be decomposed
$$\CC=(\mathcal{B}_0 \ \Box_{\pi_0} \ \mathcal{C}_0) \oplus \ \cdots \ \oplus (\mathcal{B}_s \ \Box_{\pi_s} \ \mathcal{C}_s),$$
where each $\CC_i$ is a linear code (eventually trivial) of length
$\ell$ over a the block algebra $\mathcal{B}_i$.\\[0.2ex]
{\rm b)} If ${\rm d}(\mathcal{B}_0)\leq \ldots \leq {\rm
d}(\mathcal{B}_s)$, then the minimum distance of $\CC$ is bounded
below by
    $${\rm d}(\CC)\geq \min_{0\leq i\leq s}\{{\rm d}(\CC_i)\cdot{\rm d}(\mathcal{B}_0\oplus \cdots\oplus \mathcal{B}_i)\},$$
where the minimum distance of $\CC_i$ is defined exactly as for
linear codes over fields.
\end{thm}
\begin{proof} a) This  follows directly by Theorem \ref{thm:str}.\\
b) The proof is exactly as in \cite{BZ} for linear codes over fields.
\end{proof}

\begin{rmk}{\rm In general the determination of the blocks ${\mathcal B}_i$ of $KG$ is a hard problem
and the algebraic structure may be quite complicated. If $G$ is abelian and the characteristic of $K$ does not divide $|G|$, 
then each block is explicitly known and isomorphic to a
finite extension field of $K$, by Wedderburn's Theorem. In \cite{BGSS} the authors exploit the concatenated structure of abelian quasi group codes
to find good codes. For instance, they construct a binary quasi group code for $G = C_3 \times C_3$ of index $4$ with parameters
$[36,6,16]$, which are optimal according to Grassl's list \cite{codetables}.
}
\end{rmk}

If ${\rm char}K$ does not divide $|G|$, then, by Maschke's Theorem, the algebra $KG$ is semisimple.  By Wedderburn's Theorem, we get that the blocks in \eqref{eq:decomposition} are isomorphic to full matrix rings over skew fields which are  field extensions of $K$ if $|K|$ is finite. 
The same happens, also in the non-semisimple case, for blocks of defect $0$ \cite[Chapter 16, Remark 16.3.6.]{Handbook}. 
In order to exploit the concatenation
we need an explicit isomorphism from the abstract matrix ring to $\mathcal{B}_i$. Clearly, the identity matrix is mapped to the idempotent $f_i$. Then one should look at the action of $G$ on the idempotent $f_i$, to get a representation of $G$ in the matrix ring. Note that the isomorphism is completely determined by the images of the generators of $G$. For blocks of positive defect, the isomorphism can be more complicated, but still the investigation of the structure of the Jacobson radical may help. The next examples serve as an illustration of different types of blocks and isomorphisms. We also emphasize how the concatenated construction allows to find optimal codes.

\begin{exa}\label{exa:golay}
{\rm We consider $G=\langle \alpha,\beta\mid \alpha^3=\beta^2=1,\beta\alpha\beta=\alpha^2\rangle\cong D_6$, the dihedral group of order $6$, and $K=\F_2$. Then
$$f_0=1+\alpha+\alpha^2 \text{ and } f_1=\alpha+\alpha^2$$
are the central primitive orthogonal idempotents of $KG$. Let us fix the ordering $\{1,\alpha,\alpha^2,\beta,\beta\alpha,\beta\alpha^2\}$. The the corresponding blocks $\mathcal{B}_0$ and $\mathcal{B}_1$ are the codes with generator matrices
$$B_0:=\begin{bmatrix}1&1&1&0&0&0\\
 0&0&0&1&1&1\end{bmatrix}\text{ and } B_1:=\begin{bmatrix}0&1&1&0&0&0\\
 1&1&0&0&0&0\\
 0&0&0&0&1&1\\
 0&0&0&1&1&0
 \end{bmatrix}.$$
The ring $\mathcal{B}_0$ is isomorphic to $\F_2+u\F_2$, with $u^2=0$, so that we get the map $\pi_0:\mathcal{B}_0\to \F_2^6$, 
$$1\mapsto B_{0,1},1+u\mapsto B_{0,2},$$
where $B_{0,j}$ is the $j$-th line of $B_0$. Codes over $\F_2+u\F_2$ are studied for example in \cite{AS,DS}.\\
The ring $\mathcal{B}_1$ is isomorphic to ${\rm M}_2(\F_2)$, so that we get the map $\pi_1:\mathcal{B}_1\to \F_2^6$, 
$$\begin{bmatrix}1&0\\0&1\end{bmatrix}\mapsto B_{1,1},\begin{bmatrix}1&1\\1&0\end{bmatrix}\mapsto B_{1,2}, \begin{bmatrix}0&1\\1&0\end{bmatrix}\mapsto B_{1,3}, \begin{bmatrix}1&0\\1&1\end{bmatrix}\mapsto B_{1,4},$$
where $B_{1,j}$ is the $j$-th line of $B_1$. Codes over ${\rm M}_2(\F_2)$ are studied for example in \cite{ASSY,OSB}.\\
A $G$-code $\CC\subseteq \F_2^{6\ell}$ can be then decomposed as $(\mathcal{B}_0 \ \Box_{\pi_0} \ \mathcal{C}_0) \oplus (\mathcal{B}_1 \ \Box_{\pi_1} \ \mathcal{C}_1)$, where $\CC_0\subseteq (\F_2+u\F_2)^\ell$ and $\CC_1\subseteq {\rm M}_2(\F_2)^\ell$. \\
Let us consider $\ell=4$. If $\CC_0$ is the code over $\F_2+u\F_2$ generated by
$$\begin{bmatrix}
1&0&1+u&u\\
0&1&u&1+u
\end{bmatrix}$$
and $\CC_1$ is the code over ${\rm M}_2(\F_2)$ generated by
$$\begin{bmatrix}
\begin{bmatrix}1&0\\0&1\end{bmatrix}&
\begin{bmatrix}0&0\\0&0\end{bmatrix}&
\begin{bmatrix}1&1\\0&0\end{bmatrix}&
\begin{bmatrix}0&1\\1&0\end{bmatrix}\\
\begin{bmatrix}0&0\\0&0\end{bmatrix}&
\begin{bmatrix}1&0\\0&1\end{bmatrix}&
\begin{bmatrix}1&1\\1&0\end{bmatrix}&
\begin{bmatrix}0&0\\1&0\end{bmatrix}
\end{bmatrix},$$
then the code $\CC=(\mathcal{B}_0 \ \Box_{\pi_0} \ \mathcal{C}_0) \oplus (\mathcal{B}_1 \ \Box_{\pi_1} \ \mathcal{C}_1)$ is a $[24,12,8]$ code, with generator matrix
$$G=\left[\begin{smallmatrix}
1&0&0&0&0&0&0&0&0&0&0&0&1&1&0&0&0&1&1&1&1&0&1&0\\
0&1&0&0&0&0&0&0&0&0&0&0&0&1&1&0&1&0&1&1&1&1&0&0\\
0&0&1&0&0&0&0&0&0&0&0&0&1&0&1&1&0&0&1&1&1&0&0&1\\
0&0&0&1&0&0&0&0&0&0&0&0&0&0&1&1&1&0&0&1&0&1&1&1\\
0&0&0&0&1&0&0&0&0&0&0&0&0&1&0&0&1&1&1&0&0&1&1&1\\
0&0&0&0&0&1&0&0&0&0&0&0&1&0&0&1&0&1&0&0&1&1&1&1\\
0&0&0&0&0&0&1&0&0&0&0&0&0&1&0&1&1&1&0&1&1&0&0&1\\
0&0&0&0&0&0&0&1&0&0&0&0&0&0&1&1&1&1&1&0&1&0&1&0\\
0&0&0&0&0&0&0&0&1&0&0&0&1&0&0&1&1&1&1&1&0&1&0&0\\
0&0&0&0&0&0&0&0&0&1&0&0&1&1&1&0&1&0&0&0&1&0&1&1\\
0&0&0&0&0&0&0&0&0&0&1&0&1&1&1&0&0&1&0&1&0&1&0&1\\
0&0&0&0&0&0&0&0&0&0&0&1&1&1&1&1&0&0&1&0&0&1&1&0
\end{smallmatrix}\right].$$
Note that $\CC$ is a code with the best possible minimum distance for the length $24$ and dimension $12$, by the Griesmer bound (see \cite{codetables}). Moreover, it is self-dual, hence equivalent to the extended binary Golay code.}
\end{exa}

\begin{exa}
{\rm Let us consider $G=\langle \alpha,\beta,\gamma\mid \alpha^3=\beta^2=\gamma^3=1,\beta\alpha\beta=\alpha^2, [\alpha,\gamma] =[\beta,\gamma]=1\rangle\cong S_3\times C_3$ and $K=\F_2$. Let $A=\langle \alpha\rangle=\{1,\alpha,\alpha^2\}$ and $A^\star=\{\alpha,\alpha^2\}$.
Then
$$f_0=\sum_{g\in A\cup A\gamma\cup A\gamma^2}g, \ \ f_1=\sum_{g\in A\gamma\cup A\gamma^2}g, \ \ f_2=\sum_{g\in A^\star\cup A^\star\gamma\cup A^\star\gamma^2}g \ \text{ and } \ f_3=\sum_{g\in A^\star\gamma\cup A^\star\gamma^2}g$$
are the central primitive orthogonal idempotents of $KG$. Let us fix the ordering 
$$\{1,\alpha,\alpha^2,\gamma,\alpha\gamma,\alpha^2\gamma,\gamma^2,\alpha\gamma^2,\alpha^2\gamma^2,\beta,\alpha\beta,\alpha^2\beta,\beta\gamma,\alpha\beta\gamma,\alpha^2\beta\gamma,\beta\gamma^2,\alpha\beta\gamma^2,\alpha^2\beta\gamma^2\}.$$ 
The corresponding blocks $\mathcal{B}_0$, $\mathcal{B}_1$, $\mathcal{B}_2$ and $\mathcal{B}_3$ are the codes with generator matrices
$$B_0:=\left[\begin{smallmatrix}
1&1&1&1&1&1&1&1&1&0&0&0&0&0&0&0&0&0\\
0&0&0&0&0&0&0&0&0&1&1&1&1&1&1&1&1&1
\end{smallmatrix}\right], \ \ B_1:=\left[\begin{smallmatrix}
1&1&1&0&0&0&1&1&1&0&0&0&0&0&0&0&0&0\\
0&0&0&1&1&1&1&1&1&0&0&0&0&0&0&0&0&0\\
0&0&0&0&0&0&0&0&0&1&1&1&0&0&0&1&1&1\\
0&0&0&0&0&0&0&0&0&0&0&0&1&1&1&1&1&1
\end{smallmatrix}\right],$$
$$B_2:= \left[\begin{smallmatrix}
1&0&1&1&0&1&1&0&1&0&0&0&0&0&0&0&0&0\\
0&1&1&0&1&1&0&1&1&0&0&0&0&0&0&0&0&0\\
0&0&0&0&0&0&0&0&0&1&0&1&1&0&1&1&0&1\\
0&0&0&0&0&0&0&0&0&0&1&1&0&1&1&0&1&1
\end{smallmatrix}\right], \ \ B_3:=\left[\begin{smallmatrix}
1&0&1&0&0&0&1&0&1&0&0&0&0&0&0&0&0&0\\
0&1&1&0&0&0&0&1&1&0&0&0&0&0&0&0&0&0\\
0&0&0&1&0&1&1&0&1&0&0&0&0&0&0&0&0&0\\
0&0&0&0&1&1&0&1&1&0&0&0&0&0&0&0&0&0\\
0&0&0&0&0&0&0&0&0&1&0&1&0&0&0&1&0&1\\
0&0&0&0&0&0&0&0&0&0&1&1&0&0&0&0&1&1\\
0&0&0&0&0&0&0&0&0&0&0&0&1&0&1&1&0&1\\
0&0&0&0&0&0&0&0&0&0&0&0&0&1&1&0&1&1
\end{smallmatrix}\right].$$
The ring $\mathcal{B}_0$ is isomorphic to $\F_2+u\F_2$, with $u^2=0$, so that we get the map $\pi_0:\mathcal{B}_0\to \F_2^{18}$, 
$$1\mapsto B_{0,1},1+u\mapsto B_{0,2}.$$
The ring $\mathcal{B}_1$ is isomorphic to $\F_4+u\F_4$, with $u^2=0$ and $\F_4=\F_2[\zeta]$, so that we get the map $\pi_1:\mathcal{B}_1\to \F_2^{18}$, 
$$\zeta\mapsto B_{1,1}, 1\mapsto B_{1,2}, \zeta+u\mapsto B_{1,3}, 1+u(1+\zeta)\mapsto B_{1,4}.$$
The ring $\mathcal{B}_2$ is isomorphic to ${\rm M}_2(\F_2)$, so that we get the map $\pi_2:\mathcal{B}_2\to \F_2^{18}$, 
$$\begin{bmatrix}1&1\\1&0\end{bmatrix}\mapsto B_{2,1}, \begin{bmatrix}1&0\\0&1\end{bmatrix}\mapsto B_{2,2}, \begin{bmatrix}1&0\\1&1\end{bmatrix}\mapsto B_{2,3}, \begin{bmatrix}0&1\\1&0\end{bmatrix}\mapsto B_{2,4}.$$
The ring $\mathcal{B}_3$ is isomorphic to ${\rm M}_2(\F_4)$, with $\F_4=\F_2[\zeta]$, so that we get the map $\pi_3:\mathcal{B}_3\to \F_2^{18}$, 
$$\begin{bmatrix}\zeta^2&0\\ \zeta&1\end{bmatrix}\mapsto B_{3,1}, \begin{bmatrix}\zeta&0\\0&\zeta\end{bmatrix}\mapsto B_{3,2}, \begin{bmatrix}\zeta&0\\1&\zeta^2\end{bmatrix}\mapsto B_{3,3}, \begin{bmatrix}1&0\\0&1\end{bmatrix}\mapsto B_{4,4}$$
$$\begin{bmatrix}1&1\\ \zeta&1\end{bmatrix}\mapsto B_{3,5}, \begin{bmatrix}\zeta^2&\zeta^2\\\zeta&\zeta^2\end{bmatrix}\mapsto B_{3,6}, \begin{bmatrix}\zeta^2&\zeta^2\\1&\zeta^2\end{bmatrix}\mapsto B_{3,7}, \begin{bmatrix}1&1\\0&1\end{bmatrix}\mapsto B_{4,8}.$$
\mbox{}\\
A $G$-code $\CC\subseteq \F_2^{18\ell}$ can then be  decomposed as $(\mathcal{B}_0 \ \Box_{\pi_0} \ \mathcal{C}_0) \oplus (\mathcal{B}_1 \ \Box_{\pi_1} \ \mathcal{C}_1)\oplus (\mathcal{B}_2 \ \Box_{\pi_2} \ \mathcal{C}_2)\oplus (\mathcal{B}_3 \ \Box_{\pi_3} \ \mathcal{C}_3)$, where $\CC_0\subseteq (\F_2+u\F_2)^\ell$, $\CC_1\subseteq (\F_4+u\F_4)^\ell$, $\CC_2\subseteq {\rm M}_2(\F_2)^\ell$ and $\CC_3\subseteq {\rm M}_2(\F_4)^\ell$. 
\mbox{} \\
Let us consider $\ell=3$. If $\CC_3$ is the code over ${\rm M}_2(\F_4)$ generated by
$$\begin{bmatrix}
\begin{bmatrix}1&0\\0&1\end{bmatrix}&
\begin{bmatrix}1&0\\\zeta&\zeta\end{bmatrix}&
\begin{bmatrix}0&1\\1&1\end{bmatrix}
\end{bmatrix},$$
then  $\CC= \mathcal{B}_3 \ \Box_{\pi_3} \ \mathcal{C}_3$ is a $[54,8,24]$ code. Note that this  code has the best possible minimum distance for linear codes of length $54$ and dimension $8$, by the Griesmer bound (see \cite{codetables}).
}
\end{exa}

\begin{exa}
{\rm Next we  consider $G=\langle \alpha,\beta\mid \alpha^{11}=\beta^2=1,\beta\alpha\beta=\alpha^{10}\rangle\cong D_{22}$, the dihedral group of order $22$, and $K=\F_2$. Then
$$f_0=1+\alpha+\ldots+\alpha^{10} \text{ and } f_1=\alpha+\ldots+\alpha^{10}$$
are the central primitive orthogonal idempotents of $KG$. Let us fix the ordering $$\{1,\alpha,\ldots,\alpha^{10},\beta,\beta\alpha,\ldots, \beta\alpha^{10}\}.$$ Then the corresponding blocks $\mathcal{B}_0$ and $\mathcal{B}_1$ are  codes with generator matrices
$$B_0:=\begin{bmatrix}1&\cdots&1&0&\cdots&0\\
 0&\cdots&0&1&\cdots&1\end{bmatrix}\text{ and } B_1:=\begin{bmatrix}
 0&1&\cdots&1&1&0&&\cdots&&0\\
 1&0&\cdots&1&1&0&&\cdots&&0\\
 \vdots& &\ddots&&\vdots&\vdots&&&&\vdots\\
 1&1&\cdots&0&1&0&&\cdots&&0\\
 0&&\cdots&&0&0&1&\cdots&1&1\\
 0&&\cdots&&0&1&0&\cdots&1&1\\
 \vdots&&&&\vdots&\vdots& &\ddots&&\vdots\\
 0&&\cdots&&0&1&1&\cdots&0&1
 \end{bmatrix}.$$
The ring $\mathcal{B}_0$ is isomorphic to $\F_2+u\F_2$, with $u^2=0$, so that we get the map $\pi_0:\mathcal{B}_0\to \F_2^{22}$, 
$$1\mapsto B_{0,1},1+u\mapsto B_{0,2}.$$
The ring $\mathcal{B}_1$ is isomorphic to ${\rm M}_2(\F_{32})$, where $\F_{32}=\F_2[\zeta]$ with $\zeta^5+\zeta^2+1=0$, so that we get the map $\pi_1:\mathcal{B}_1\to \F_2^{22}$, 
$$\begin{bmatrix}1&0\\0&1\end{bmatrix}\mapsto B_{1,1},\begin{bmatrix}1&\zeta^9\\\zeta^9&\zeta\end{bmatrix}\mapsto B_{1,2}, \begin{bmatrix}0&1\\1&0\end{bmatrix}\mapsto B_{1,11}.$$
A $G$-code $\CC\subseteq \F_2^{22\ell}$ can then be  decomposed as $(\mathcal{B}_0 \ \Box_{\pi_0} \ \mathcal{C}_0) \oplus (\mathcal{B}_1 \ \Box_{\pi_1} \ \mathcal{C}_1)$, where $\CC_0\subseteq (\F_2+u\F_2)^\ell$ and $\CC_1\subseteq {\rm M}_2(\F_{32})^\ell$.\\
Let us consider $\ell=5$. If  $\CC_1$ is the code over ${\rm M}_2(\F_{32})$ generated by
$$\begin{bmatrix}
\begin{bmatrix}
\zeta^{18}&1\\
\zeta^{9}&\zeta^{21}
\end{bmatrix}&
\begin{bmatrix}
\zeta^{15}&\zeta^{5}\\
\zeta^{23}&\zeta^{12}
\end{bmatrix}&
\begin{bmatrix}
\zeta&\zeta^{28}\\
\zeta^{20}&\zeta^{24}
\end{bmatrix}&
\begin{bmatrix}
\zeta^{25}&\zeta^{20}\\
\zeta^{11}&\zeta^{15}
\end{bmatrix}&
\begin{bmatrix}
\zeta^{19}&\zeta^{28}\\
\zeta&\zeta^{18}
\end{bmatrix}\\
\begin{bmatrix}
\zeta^{25}&\zeta^{27}\\
\zeta^{29}&\zeta^{3}
\end{bmatrix}&
\begin{bmatrix}
\zeta^{3}&\zeta^{2}\\
\zeta^{5}&\zeta^{19}
\end{bmatrix}&
\begin{bmatrix}
\zeta^{26}&\zeta^{26}\\
\zeta^{16}&\zeta^{9}
\end{bmatrix}&
\begin{bmatrix}
\zeta^{23}&\zeta^{23}\\
\zeta^{15}&\zeta^{8}
\end{bmatrix}&
\begin{bmatrix}
\zeta^{23}&\zeta^{17}\\
\zeta&\zeta^{7}\\
\end{bmatrix}
\end{bmatrix},$$
then the code $\CC=\mathcal{B}_1 \ \Box_{\pi_1} \ \mathcal{C}_1$ is a $[110,40,22]$ code.
Note that the best known minimum distance of a linear binary $[110,40]$ code is $24$ (see \cite{codetables}).}
\end{exa}

\bigskip

\section{Self-duality of quasi group codes}

  Let $K$ be a finite field and let $G$ be a finite
group of order $n$. For $\ell\in \mathbb{N}$, there is a Euclidean bilinear form  on $KG^\ell = KG \oplus \cdots \oplus KG$ 
which is defined by
$$\langle \sum_{i=1}^\ell a_i, \sum_{i=1}^\ell b_i \rangle =
\sum_{i=1}^\ell  \varphi(a_i)\cdot \varphi(b_i)$$ where $\varphi$ is
the isomorphism defined in \eqref{eq:iso} and  $\varphi(a_i)\cdot
\varphi(b_i)$ is the standard inner product on $K^n$. Thus $KG^\ell
= KG \perp \cdots \perp KG$.

For a linear code $\mathcal{C}$ over $K$ of length $n\ell$, the
\emph{dual code} is classically defined as the linear code
$\mathcal{C}^\perp=\{v\in K^{n\ell} \mid v\cdot c=0 \text{ for all }
c\in \mathcal{C}\}$. The definition is the same for a quasi-$G$
code, but it can be formulated also in terms of the above Euclidean
bilinear form: if $\mathcal{C}$ is a $KG$-submodule of $KG^\ell$, then
$$\mathcal{C}^\perp=\{v\in KG^\ell \mid \langle v,c\rangle=0\text{ for all }
c\in \mathcal{C}\}.$$ The two definitions coincide via the
isomorphism $\varphi^\ell$. In both cases, $\mathcal{C}$ is called
\emph{self-dual} if $\mathcal{C}=\mathcal{C}^\perp$.

However, another notion of duality is used in representation theory, as we
 see in the following definition.

\begin{dfn} \label{dualmodule} {\rm Let $\mathcal{V}$ be a (right) $KG$-module. The vector space
$\ho_K(\mathcal{V},K)$ of all $K$-linear maps  from $\mathcal{V}$ to
$K$ becomes a right $KG$-module by
$$ (\alpha g)v = \alpha(vg^{-1}) $$
for $v \in \mathcal{V}, g \in G$ and $\alpha \in
\ho_K(\mathcal{V},K)$. This module is denoted by $\mathcal{V}^*$ and
called the \emph{dual module} of $\mathcal{V}$. If $\mathcal{V}\cong
\mathcal{V}^*$, we say that $\mathcal{V}$ is a \emph{self-dual
$KG$-module}. }
\end{dfn}

Observe that the trivial $KG$-module $K_G$ is always self-dual. Furthermore, 
for any $KG$-module $\mathcal{V}$, we have $\dim \mathcal{V}=\dim \mathcal{V}^*$.\\

In \cite{Wil} it has been shown that a self-dual group code in $KG$
exists if and only if the characteristic of $K$ is 2 and $|G|$ is
even. For quasi group codes we have the following generalization.
Note that for finite fields of odd cardinality the group $G$ does
not play any role.

\begin{thm}\label{thm:sd} For any finite group $G$ there exists a self-dual quasi-$G$ code over $K$ of index $\ell$ if and only if
one of the following holds true.
\begin{itemize}
\item[(i)] $|K| \equiv 1 \bmod 4$ and $ 2 \mid \ell$.
\item[(ii)] $|K| \equiv 3 \bmod 4$ and $ 4 \mid \ell$.
\item[(iii)] $|K|$ is even and $ 2 \mid \ell$ or $ 2 \mid |G|$.
\end{itemize}
\end{thm}
\begin{proof}  Suppose that $|K|$ is odd. Let $\mathcal{C} \leq KG^\ell =: \mathcal{V}$ be a quasi group
 code and suppose that $\mathcal{C}=\mathcal{C}^\perp$
is self-dual. We argue now similar as in \cite{Wil}. First note that
$\mathcal{V}/\mathcal{C} = \mathcal{V}/\mathcal{C}^\perp \cong
\mathcal{C}^*$ as $KG$-modules. Thus the multiplicity of the trivial
$KG$-module $K_G$ as a composition factor of $\mathcal{V}$ is even,
since $K_G$ is a self-dual irreducible $KG$-module. In particular,
if $T$ is a Sylow $2$-subgroup of $G$, then the multiplicity of the
trivial $KT$-module $K_T$ in the restriction $\mathcal{V}|_T$ is
even. On the other hand, by Maschke's Theorem, the multiplicity of
$K_T$ in $KT$ is one. Thus the multiplicity of $K_T$ in $KG|_T$ is
$|G:T|$. It follows that the multiplicity of $K_T$ in
$\mathcal{V}|_T$ is $\ell|G:T|$. Since $|G:T|$ is odd, we see that
$\ell$ must be even. \\[2ex]
(i) By the above we only have to show that $KG^2$
contains a self-dual quasi group code. Since $|K| \equiv 1 \bmod 4$
there exists $x \in K$ such that $x^2 = -1$. Now we consider the
$KG$-module
$$ \mathcal{C} = \{a \oplus xa \mid a \in KG \} \leq KG^2.$$
Clearly $\dim \mathcal{C} = |G| = \frac{\dim \mathcal{V}}{2}$.
Furthermore, since
$$ \langle a \oplus xa, b \oplus xb \rangle = (a,b) + (xa,xb) = (a,b) + x^2(a,b)= 0,$$
the quasi group code $\mathcal{C}$ is self-dual. For $\ell\geq 2$,
it is enough to take direct sums of this code.
\\[2ex]
(ii) Suppose that there exists a self-dual quasi-$G$ code of index
$\ell$. As shown in the first paragraph of the proof
 we have $\ell=2m$.
Let $T$ be a Sylow $2$-subgroup of $G$. We consider $\mathcal{V}|_T$
which is a semi-simple $KT$-module. Note that the maximal submodule
$\mathcal{M}$ of $\mathcal{V}|_T$ on which $T$ acts trivially has
dimension $|G:T|\ell$. The Gram matrix ${\rm G}(\mathcal{M})$ of the
form restricted to $\mathcal{M}$ is a diagonal matrix of type
$(|G:T|\ell, |G:T|\ell)$ with entries $|T|$ in the diagonal. On the
other hand, $\mathcal{C}$ must intersect $\mathcal{M}$ in a totally
isotropic subspace of dimension $|G:T|m$. By (\cite{HW}, Satz
7.3.12), we get $|T|^{|G:T|\ell} (-1)^{|G:T|m}=\det {\rm
G}(\mathcal{M})(-1)^{|G:T|m} = (-1)^m \in K^{*2}$. This forces $2
\mid m$ since $|G| \equiv 3 \bmod 4$. Thus $4 \mid \ell$.

To prove the converse we only have to show that $KG^4$ contains a
quasi group code. We choose $x$ and $y$ in $K$ such that $x^2 + y^2
=-1$. Such elements exist since $\{-1 -x^2 \mid x \in K\}$ and
$\{y^2 \mid y \in K\}$ are two sets of cardinality $(|K|+1)/2$, so
that their intersection is non-empty. Next we put
$$\mathcal{C} = \{ (x a,y a,a,0), (0,-b,y b, x b) \mid a,b \in KG\} \leq KG^4.$$
Clearly, $\mathcal{C} $ is a $KG$-module of dimension $|G|^2
=\frac{\dim V}{2}$ and $\mathcal{C}  \subseteq \mathcal{C} ^\perp$.
 Thus $\mathcal{C} $ is a self-dual quasi group code. For $\ell\geq 4$,
it is enough to take direct sums of this code.\\[2ex]
(iii) Suppose that $\mathcal{C} = \mathcal{C}^\perp \in KG^\ell=:\mathcal{V}$ where $\ell$ is odd. Thus $2
\mid \dim \mathcal{V} = \ell|G|$, hence $2 \mid |G|$.

Conversely, if $\ell$ is odd and $|G|$ is even the existence of
$\mathcal{C}=\mathcal{C}^\perp \in KG^\ell$ follows immediately from
\cite{Wil} since $KG^\ell = KG \perp \cdots \perp KG$. For $\ell$
even we copy the proof of (i) with $x=1$.
\end{proof}

\begin{rmk}{\rm Let $K$ be a finite field such that $|K| \equiv 3 \bmod 4$. Already in the early paper \cite{Pl}, it is proved
that there exists a self-dual code in $K^n$ if and only if $4 \mid
n$. }
\end{rmk}

\begin{rmk} {\rm The ternary $[12,6,6]_3$ self-dual extended Golay code is not a group code. According to \cite{BLM},
it is a right ideal in a twisted group algebra $\F_3^\alpha A_4$ where $A_4$ denotes the alternating group on $4$ letters.
Actually, it is also a self-dual quasi
group code of index $\ell=4$ for a cyclic group of order $3$.
}
\end{rmk}

\bigskip

\bibliographystyle{alpha}

\end{document}